\documentclass[reqno]{amsart}
\usepackage[foot]{amsaddr}
\usepackage{graphicx}
\graphicspath{ {./figures/} }
\usepackage[margin=3cm]{geometry}
\usepackage{amsmath, amssymb,amsthm}
\usepackage{cite}
\usepackage[scr=rsfs]{mathalpha}
\usepackage[shortlabels]{enumitem}
\usepackage{mlmodern}
\DeclareSymbolFont{largesymbols}{OMX}{cmex}{m}{n} 
\usepackage{mathtools} 
\usepackage{stmaryrd} 
\usepackage{tikz}
\usetikzlibrary{decorations.pathreplacing,patterns,math,external,decorations.pathmorphing,arrows.meta}
\usepackage{upref} 
\usepackage[colorlinks]{hyperref}
\hypersetup{citecolor=blue,filecolor=blue,linkcolor=blue,urlcolor=navyblue}
\definecolor{navyblue}{rgb}{0.0, 0.0, 0.5}
\usepackage{CJKutf8}
\DeclareRobustCommand{\rc}[1]{\begin{CJK*}{UTF8}{gbsn}#1\end{CJK*}}

\newtheorem{thm}{Theorem}[section]
\newtheorem{prop}[thm]{Proposition}

\newtheorem{cor}[thm]{Corollary}
\theoremstyle{definition}
\newtheorem{defn}{Definition}
\newtheorem{rmk}{Remark}[section]

\newtheorem*{claim*}{Claim}

\newcommand{\N}{\mathbb{N}}
\newcommand{\Z}{\mathbb{Z}}
\newcommand{\R}{\mathbb{R}}
\newcommand{\C}{\mathbb{C}}

\newcommand{\Sb}{\mathbb{S}}

\renewcommand{\Im}{\operatorname{\mathrm{Im}}}
\renewcommand{\Re}{\operatorname{\mathrm{Re}}}

\newcommand{\E}{\mathbf{E}}
\newcommand{\oneb}{\mathbf{1}}
\renewcommand{\P}{\mathbf{P}}

\newcommand{\CN}{\mathcal{CN}}
\newcommand{\RN}{\mathcal{N}}

\newcommand{\dsim}{\sim}
\newcommand{\rvsim}{\overset{d}{=}}

\newcommand{\Tr}{\operatorname{Tr}}

\newcommand{\hs}{\mathrm{hs}}

\newcommand{\U}{\mathrm{U}}
\newcommand{\Sp}{\mathrm{Sp}}

\newcommand{\sq}{\mathbf{s}}

\newcommand{\Haf}{\operatorname{Haf}}
\newcommand{\Per}{\operatorname{Per}}

\newcommand{\sech}{\operatorname{sech}}
\newcommand{\op}[1]{\operatorname{#1}}

\newcommand\numberthis{\stepcounter{equation}\tag{\theequation}}
\numberwithin{equation}{section}

\begin{document}

\title[]{Anticoncentration and entanglement in Gaussian boson sampling}

\author{Laura Shou$^{1,2}$}
\author{Adam Ehrenberg$^{1,2}$}
\author{Yu-Xin Wang (\rc{王语馨})$^2$}
\author{Joseph T. Iosue$^{1,2,3}$}
\author{Alexey V. Gorshkov$^{1,2}$}

\address{\normalfont$^1$Joint Quantum Institute, Department of Physics, NIST/University of Maryland, College Park, MD 20742, USA}

\address{\normalfont$^2$Joint Center for Quantum Information and Computer Science, NIST/University of Maryland, College Park, MD, 20742, USA}

\address{\normalfont$^3$Present Address: Microsoft Quantum, Redmond, WA 98052, USA}

\begin{abstract}
Anticoncentration and entanglement are two properties used to study classical hardness or easiness of Gaussian boson sampling in varying regimes. In this paper, we consider three directions concerning anticoncentration and entanglement in Gaussian boson sampling: 
(1) We derive closed-form expressions for the second moment of hafnians of symmetric Gaussian products, and use this to precisely locate the anticoncentration transition as a function of the number of squeezed input modes.
(2) We derive closed-form expressions for the R\'enyi-$\alpha$ Page curves for Gaussian boson sampling.
(3) We study unequal input squeezing parameters $(s_i)_i$, and prove estimates as well as monotonicity of the average-case R\'enyi-2 entropy Page curve in terms of the magnitudes $|s_i|$, which allow for extending equal squeezing results to unequal squeezing.
\end{abstract}

\maketitle

\section{Introduction}

Quantum random sampling tasks, such as random circuit sampling \cite{boixo2018characterizing,bouland2019complexity} and boson sampling \cite{aa,hamilton2017gaussian}, are proposed as routes to demonstrating quantum advantage in currently realizable experimental systems; see \cite{hangleiter2023computational} for a review.
In these sampling tasks, anticoncentration of the output probabilities and entanglement of the system are used to theoretically analyze regimes of classical hardness or easiness of sampling.
For example, anticoncentration can be used to provide evidence for hardness of approximate sampling, while limited entanglement can allow for classically efficient matrix product state simulations \cite{vidal2003efficient,oh2024classical}.
These properties have been studied in many regimes, including for Fock or Gaussian boson sampling, e.g. in \cite{ehrenberg2025transition,ehrenberg2025second,matinez2024linear,mhiri2026boson,kolarovszki2026general} (anticoncentration) and \cite{serafini2007canonical,serafini2007teleportation,Fukuda2019Typical-entangl,iosue2023page,youm2025average,shou2026entanglement,zhao2026weak} (entanglement), among others.

We consider a photonic system with $m$ modes and an initial separable squeezed Gaussian state on these modes, with squeezing parameters $(s_i)_{i=1}^m$. For simplicity, we typically take $k$ modes with equal squeezing parameter $s>0$, and the remaining $m-k$ modes as vacuum, though, in the last Section~\ref{sec:unequal}, we will relax this condition and consider general $(s_i)_{i=1}^m$.
In a Gaussian boson sampling experiment, the initial Gaussian state is evolved through a network of beamsplitters and phaseshifters, which can be described by an $m\times m$ linear optical unitary matrix $U$. The resulting Gaussian state at the end is measured in the photon number basis, producing a vector $\mathbf n\in\{0,1,2,\ldots\}^m$ of photon counts, with total photon number $N\equiv 2n:=\sum_{j=1}^m\mathbf n_j$ (note that $N$ is even because the squeezing process only produces superpositions of Fock states with even counts). 
The output $\mathbf n$ is said to be collision-free if $\mathbf n\in\{0,1\}^m$. For such outputs, and in the case of $k$ equally squeezed modes with squeezing parameter $s>0$, the probability of observing photon count $\mathbf n$ is \cite{hamilton2017gaussian,kruse2019detailed}
\begin{align}\label{eqn:prob}
\Pr[\mathbf n]&=\frac{\tanh^{2n} s}{\cosh^k s}|[\Haf(UI_kU^T)_{\mathbf n,\mathbf n}]|^2,
\end{align}
where $I_k$ is the diagonal matrix consisting of $k$ 1s and $m-k$ 0s, $(UI_kU^T)_{\mathbf n,\mathbf n}$ denotes the $N\times N$ submatrix of $(UI_kU^T)$ obtained by keeping only the rows and columns corresponding to 1s in $\mathbf n$, and $\Haf[X]$ denotes the \emph{hafnian} of a symmetric matrix $X$,
\begin{align}
\Haf[X]&=\sum_{\pi\in\mathcal P_2(2n)}\prod_{\{i,j\}\in\pi}X_{ij},
\end{align}
for $\mathcal P_2(2n)$ the set of all perfect matchings of $2n$ elements.

Anticoncentration of the output probabilities refers to the property that most output probabilities $\P[\mathbf n]$ are roughly uniform over the random unitary ensemble, which intuitively suggests that classical sampling may be difficult, since one has to consider all possible outputs, rather than just a few if the output probabilities were concentrated on certain outputs \cite{aa}. More concretely, weak anticoncentration as measured using second moments and the Paley--Zygmund inequality can be used to give evidence for classical hardness of approximate average-case sampling \cite{hangleiter2023computational}.

The initial squeezed Gaussian input state is not entangled across the input modes, but may become entangled after evolving through the linear optical unitary $U$. For Haar random $U$, the average-case entanglement for the resulting state can be expressed through the R\'enyi-$\alpha$ Page curves \cite{page1993average,iosue2023page,youm2025average}. For simplicity, we first consider all input modes equally squeezed with squeezing parameter $s>0$. For average-case R\'enyi-$\alpha$ entropy $S_\alpha(U;s,r)$ across a subsystem cut of size $k=rm$, one can define the $m\to\infty$ Page curve $\lim_{m\to\infty}\frac1m\E S_\alpha(U;s,r)$. This describes the resulting entanglement across different subsystem sizes for the Haar random unitary setting.

In this paper, we address several questions concerning anticoncentration and entanglement in Gaussian boson sampling. 
\begin{enumerate}

\item First, we derive closed-form expressions for the second moments of hafnians of the symmetric product of $k\times 2n$ Gaussian random matrices, and use this to analytically locate the transition as a function of the number of squeezed input modes $k$ between weak anticoncentration and lack of anticoncentration that was originally studied numerically in \cite{ehrenberg2025transition,ehrenberg2025second}. We find the anticoncentration transition occurs at $k\propto n^2/\log n$. (Theorem~\ref{thm:moments}.)

\item Next, we use random matrix resolvent results to derive closed-form expressions for the R\'enyi-$\alpha$ Page curves, for $\alpha\ge2$ an integer, in the limit as the number of modes $m\to\infty$. (Theorem~\ref{thm:s2} and Corollary~\ref{cor:alpha}.) These agree with the infinite series expansions derived in \cite{iosue2023page,youm2025average}. 
We also give an alternate derivation of the first subleading (constant order) term of the R\'enyi-$\alpha$ Page curve.

\item Finally, we consider entanglement with unequal input squeezing parameters $(s_i)_i$. We prove monotonicity of average-case R\'enyi-2 entropy as a function of the squeezing magnitudes $|s_i|$ (Theorem~\ref{thm:monotonicity}), which was conjectured in \cite{iosue2023page}. 
This allows us to characterize entanglement in the broad setting of an arbitrary initial product squeezed state. 
The monotonicity also applies to other random ensembles including depth-$d$ random linear optical networks with independent Haar random beamsplitter-phaseshifter combinations. We also prove separate bounds for the case of $L$ modes equally squeezed and the remaining $m-L$ modes vacuum (Theorem~\ref{thm:unequal}).

\end{enumerate}

\subsection{Outline and notation}
The rest of the paper is organized as follows.
\begin{itemize}
\item In Section~\ref{sec:m2}, we study the anticoncentration transition and hafnian moments, proving Theorem~\ref{thm:moments}.
\item In Section~\ref{sec:page}, we derive the closed-form expressions for the R\'enyi-$\alpha$ Page curves, proving Theorem~\ref{thm:s2} and Corollary~\ref{cor:alpha}.
\item In Section~\ref{sec:unequal}, we study entanglement with unequal squeezing parameters, proving Theorems~\ref{thm:monotonicity} and \ref{thm:unequal}.
\end{itemize}
Throughout the paper, generic constants $c,C$ may change from line to line.
In all sections, the number of modes is $m$. In Section~\ref{sec:m2}, we also consider $1\le k\le m$ the number of initially squeezed modes, and $2n$, the observed photon count.
In Sections~\ref{sec:page} and \ref{sec:unequal}, which consider entanglement entropies, we will use $k$ to refer to the size of the subsystem $\Gamma$, $|\Gamma|=k$. This $k$ is entirely unrelated to the $k$ in Section~\ref{sec:m2}, but maintains consistency with some of the cited entanglement works.

\section{Anticoncentration transition and hafnian moments}\label{sec:m2}

In this section, we prove Theorem~\ref{thm:moments} on anticoncentration and moments for random hafnians in Gaussian boson sampling. This analytically proves the precise location of the anticoncentration transition observed in \cite{ehrenberg2025transition,ehrenberg2025second}.
In this setting, we consider Haar random linear optical unitaries $U$, initially separable squeezed states with equal squeezing parameter $s>0$ on $k$ of $m$ modes, and work in the dilute regime where the expected photon number $\E[2n]=k\sinh^2(s)$ is $o(\sqrt{m})$, which gives large probability for photon count outcomes to be collision-free.
The observed photon number is $N=2n$.
From the hiding property \cite{deshpande2022quantum,shou2026proof,hiding2}, we know for sufficiently small $n$, such as $n=o(\sqrt{k})$ or $nk=o(m)$, the submatrices $(UI_kU^T)_{\mathbf n,\mathbf n}$ in the output probabilities \eqref{eqn:prob} are well-approximated in total variation distance (TVD) by a suitably normalized symmetric product of standard complex \emph{Gaussian} matrices $X$. 
The hiding property is also conjectured to hold for further $n$ range when $k=o(m)$ \cite{deshpande2022quantum,ehrenberg2025transition}.
For these reasons, as in \cite{ehrenberg2025transition,ehrenberg2025second}, we work in the dilute regime and study Gaussian matrices $X$, which have independent entries, in place of the linear optical unitary $U$. (Note however that TVD does not directly transfer moment asymptotics; see \cite[\S A.2]{ehrenberg2025second} for discussion on indirect anticoncentration transfer, and Remark~\ref{rmk:unitary} for unitary hafnian moments \cite{mhiri2026boson,kolarovszki2026general}.)

For $X$ a $k\times 2n$ matrix of i.i.d. standard complex Gaussians, $X_{ij}\dsim\CN(0,1)$, let
\begin{align}
M_t(k,n):=\E[|\Haf(X^TX)|^{2t}],\quad\text{and}\quad m_t(k,n):={M_1(k,n)^t}/{M_t(k,n)}.
\end{align}
\begin{defn}
We say there is weak anticoncentration if $m_2(k,n)\ge n^{-O(1)}$, and lack of anticoncentration if $m_2(k,n)=n^{-\omega(1)}$.
\end{defn}

We give a closed-form expression for $m_2(k,n)$ and use this to prove that the precise location of the anticoncentration transition studied in \cite{ehrenberg2025transition,ehrenberg2025second} is at $k\propto n^2/\log n$.

\begin{thm}[hafnian moments and anticoncentration transition]\label{thm:moments}
Let $N=2n$, and let $M_t(k,n):=\E|\Haf(X^TX)|^{2t}$ for $X$ a $k\times 2n$ matrix of i.i.d. standard complex Gaussians. 
\begin{enumerate}[(i)]
\item For any $t\in\N$,
\begin{align}\label{eqn:haf-per}
M_t(k,n)&=\E_{G,H}[\Per(G^TH)^N],
\end{align}
for $G,H\in\R^{k\times t}$ i.i.d. standard real Gaussian matrices. Note that $G^TH$ is a $t\times t$ matrix.

\item For $t=2$, the above evaluates to
\begin{align}\label{eqn:M2}
M_2(k,n)&=(2n)!2^{2n}[(k/2)_n]^2\sum_{r=0}^n\binom{n}{r}^2\frac{(1/2)_r}{(k/2)_r},
\end{align}
where $(x)_r=\prod_{j=0}^{r-1}(x+j)$ is the Pochhammer symbol/rising factorial.

\item There is weak anticoncentration for $k=\Omega(n^2/\log n)$, and no weak anticoncentration for $k=o(n^2/\log n)$. Additionally, if $k/n^2\to\infty$, then
\begin{align}\label{eqn:m2-limit}
m_2(k,n)=\frac{1+o(1)}{\sqrt{\pi n}},
\end{align}
as $n\to\infty$.
\end{enumerate}
\end{thm}

\begin{rmk}\label{rmk:unitary}
One should be able to show leading order agreement of the above with the moments for actual unitary matrices in the collision-free/dilute regime, for example using asymptotic Weingarten calculus \cite{CollinsSniady2006}. 
However, we note that recent works \cite{mhiri2026boson,kolarovszki2026general} study the exact unitary case directly using a representation theoretic approach, which they use to derive closed-form expressions and anticoncentration for boson sampling with the exact linear optical unitary, in both the collision-free/dilute and saturated regimes. In particular, for Gaussian boson sampling, \cite[Proposition 5, Eq.~(100)]{kolarovszki2026general} gives anticoncentration bounds and locations which apply in arbitrary regimes for the exact linear optical unitary, for example implying a transition much closer to $k\approx m$ in the saturated regime $n=\Theta(m)$.
\end{rmk}

\begin{proof}[Proof of Theorem~\ref{thm:moments}(i)]
To evaluate $M_t(k,n)=\E|\Haf(X^TX)|^{2t}$, the main observation is that the hafnian simplifies by first applying Wick's/Isserlis' theorem, which states that for $Y=(Y_1,\ldots,Y_m)$ a centered multivariate complex or real Gaussian,
\begin{align}\label{eqn:wick}
\E[Y_1\cdots Y_m]&=\sum_{\pi\in\mathcal P_2(m)}\prod_{\{i,j\}\in\pi}\E[Y_iY_j]\equiv\Haf(\Sigma),
\end{align}
where $\Sigma$ is the $m\times m$ covariance matrix given by $\Sigma_{ij}=\E[Y_iY_j]$. 
For any fixed $k\times N$ matrix $X$, letting $X_i$ denote the $i$th column of $X$, the matrix $X^TX$ is seen to be the covariance matrix of $Y=(X_i\cdot g)_{i=1}^N$ for $g\sim\RN(0,I_k)$ real multivariate Gaussian.
Since the vector $Y$ is multivariate Gaussian for any fixed $X$, Wick's/Isserlis' theorem \eqref{eqn:wick} gives
\begin{align}
\Haf(X^TX)&=\E_g\prod_{i=1}^N(X_i\cdot g).
\end{align}
Let $g_1,\ldots,g_t,h_1,\ldots,h_t\dsim\RN(0,I_k)$ be i.i.d., and let $G$ be the $k\times t$ matrix whose columns are $g_1,\ldots,g_t$, and $H$ the $k\times t$ matrix whose columns are $h_1,\ldots,h_t$.
Recalling that all entries of $X$ are independent standard complex Gaussian, we can expand and use Wick's/Isserlis' theorem (in the normal usage direction) to get
\begin{align*}
\E|\Haf(X^TX)|^{2t}&=\E_X\E_{G,H}\prod_{i=1}^N(X_i\cdot g_1)\cdots (X_i\cdot g_t)(\bar X_i\cdot h_1)\cdots(\bar X_i\cdot h_t)\\
&=\E_{G,H}\prod_{i=1}^N\sum_{\substack{j_1,\ldots,j_t=1\\l_1,\ldots,l_t=1}}^k\E_{X_i}[(X_i)_{j_1}(g_1)_{j_1}\cdots(X_i)_{j_t}(g_t)_{j_t}(\bar X_i)_{l_1}(h_1)_{l_1}\cdots(\bar X_i)_{l_t}(h_t)_{l_t}]\\
&=\E_{G,H}\prod_{i=1}^N\sum_{\substack{j_1,\ldots,j_t=1\\l_1,\ldots,l_t=1}}^k
(g_1)_{j_1}\cdots (g_t)_{j_t}(h_1)_{l_1}\cdots(h_t)_{l_t}\sum_{\sigma\in S_t}
\prod_{m=1}^t\E_{X_i}[(X_i)_{j_m}(\bar X_i)_{l_{\sigma(m)}}]\\
&=\E_{G,H}\prod_{i=1}^N\sum_{\sigma\in S_t}(g_1\cdot h_{\sigma(1)})\cdots(g_t\cdot h_{\sigma(t)})=\E_{G,H}[\Per(G^TH)^N], \numberthis
\end{align*}
which is \eqref{eqn:haf-per}.
Note that in applying Wick's/Isserlis' theorem in the second to last line, we used that we must pair each non-conjugated $(X_i)_j$ with a conjugated $(\bar X_i)_j$ to obtain a nonzero term, which gives a sum over $\sigma\in S_t$ instead of over all possible pairings of $[2t]$.
\end{proof}

\begin{proof}[Proof of Theorem~\ref{thm:moments}(ii)]
From part (i), we have for $g,g',h,h'$ i.i.d. $\RN(0,I_k)$ real multivariate Gaussians,
\begin{align}
M_2(k,n)=\E|\Haf(X^TX)|^4&=\E_{g,g',h,h'}[(g\cdot h)(g'\cdot h')+(g\cdot h')(g'\cdot h)]^N.
\end{align}
Condition on (fix) $g,g'$. We will evaluate the moment-generating function (mgf) of $A:=(g\cdot h)(g'\cdot h')+(g\cdot h')(g'\cdot h)$, and then extract the $N$th moment of $A$.
We first rewrite $A$ in a more convenient form, as the inner product of two independent Gaussian random vectors. Let 
$x:=(g\cdot h,g'\cdot h)$, and $y:=(g\cdot h',g'\cdot h')$, which are independent centered real Gaussian vectors with covariance
\begin{align*}
\Sigma=\begin{pmatrix}\|g\|_2^2&g\cdot g'\\g\cdot g'&\|g'\|_2^2\end{pmatrix}.
\end{align*}
We see $A=\langle x|J|y\rangle$ for $J=\begin{pmatrix}0&1\\1&0\end{pmatrix}$. Since $x,y$ are independent Gaussians, we can use the $\RN(0,\Sigma)$ multivariate Gaussian mgf $\E[e^{\mathbf t^TZ}]=e^{\frac12\mathbf t^T\Sigma \mathbf t}$, followed by the mgf for quadratic forms of Gaussians (or just the Gaussian integral again), to evaluate
\begin{align*}
\E e^{t\langle x|J|y\rangle}=\E_x\E_ye^{t\langle x|J|y\rangle}&=\E_xe^{\frac12t^2x^TJ\Sigma Jx}
=\det(I-t^2J\Sigma J\Sigma)^{-1/2}.\numberthis\label{eqn:A-mgf}
\end{align*}
The eigenvalues of 
\begin{align*}
\Sigma J\Sigma J=\begin{pmatrix}(g\cdot g')^2+\|g\|_2^2\|g'\|_2^2&2(g\cdot g')\|g\|_2^2\\
2(g\cdot g')\|g'\|_2^2&(g\cdot g')^2+\|g\|_2^2\|g'\|_2^2\end{pmatrix}
\end{align*}
are $(g\cdot g'\pm\|g\|_2\|g'\|_2)^2$, which with \eqref{eqn:A-mgf} gives
\begin{align}\label{eqn:mgf}
\E[e^{tA}|g,g']&=\frac1{\sqrt{(1-t^2[\|g\|_2\|g'\|_2+ g\cdot g']^2)(1-t^2[\|g\|_2\|g'\|_2- g\cdot g']^2)}}.
\end{align}
Letting $\xi_N[g,g']$ denote the coefficient on $t^N$ in the power series expansion of $t\mapsto\E[e^{tA}|g,g']$, we have
\begin{align*}
\E[(g\cdot h)(g'\cdot h')+(g\cdot h')(g'\cdot h)]^N&=N!\,\E\xi_N[g,g'].
\end{align*}
To simplify $\E\xi_N[g,g']$, let $w:=t^2\|g\|_2^2\|g'\|_2^2$ and $q:=\frac{g}{\|g\|_2}\cdot\frac{g'}{\|g'\|_2}$. Then \eqref{eqn:mgf} can be written as ${1}/{\sqrt{(1-w(1+q)^2)(1-w(1-q)^2)}}$, and the relation between $w$ and $t$ implies
\begin{align}\label{eqn:xi}
\xi_N[g,g']&=\|g\|_2^N\|g'\|_2^N \zeta_{n}(q),
\end{align}
recalling $N=2n$, and where $\zeta_{n}(q)$ denotes the coefficient on $w^{n}$ in the power series expansion of $1/{\sqrt{(1-w(1+q)^2)(1-w(1-q)^2)}}$.
Since $q:=\frac{g}{\|g\|_2}\cdot\frac{g'}{\|g'\|_2}$ is independent of the vector norms $\|g\|_2$ and $\|g'\|_2$, we then get
\begin{align}\label{eqn:m2-zeta}
M_2(k,n)&=N!\E_g\|g\|_2^N\,\E_{g'}\|g'\|_2^N\,\E_q \zeta_{n}(q).
\end{align}
We have $\E\|g\|_2^N=2^n\frac{\Gamma(n+\frac{k}{2})}{\Gamma(\frac{k}{2})}=2^n(k/2)_n$, since it is the $n$th moment of a $\chi^2$ random variable with $k$ degrees of freedom.
To determine the coefficient of $w^n$, we could directly do a power series expansion; however, the following method ends up being nicer.
Using the identity $\frac{1}{\sqrt{a^2-b^2}}=\frac{1}{2\pi}\int_0^{2\pi}\frac{d\theta}{a+b\cos\theta}$ 
for $a>|b|$, with $a=1-w(1+q^2)$ and $b=-2wq$, take $w\in\R$ small, recall $q\in\R$, and rewrite 
\begin{align*}
\frac{1}{\sqrt{(1-w(1+q)^2)(1-w(1-q)^2)}}&=\frac1{2\pi}\int_0^{2\pi}\frac{d\theta}{1-w|1+qe^{i\theta}|^2}\\
&=\frac{1}{2\pi}\int_0^{2\pi}\sum_{j=0}^\infty w^j|1+qe^{i\theta}|^{2j}\,d\theta.\numberthis
\end{align*}
But now we can read off the coefficient of $w^n$ (since everything is convergent for small $w$) to be
\begin{align*}
\zeta_n(q)&=\frac1{2\pi}\int_0^{2\pi}(1+qe^{i\theta})^n(1+qe^{-i\theta})^n\,d\theta\\
&=\frac1{2\pi}\int_0^{2\pi}\sum_{j,\ell=0}^n\binom{n}{j}\binom{n}{\ell}q^{j+\ell}e^{i(j-\ell)\theta}\,d\theta=\sum_{j=0}^n\binom{n}{j}^2q^{2j}.\numberthis
\end{align*}
Finally, since $q\rvsim\langle u,v\rangle\rvsim u_1$ for $u,v$ independent random real unit vectors in $\Sb_\R^{k-1}$, we see $q^2\dsim\op{Beta}(1/2,(k-1)/2)$ for $k\ge2$, e.g. writing $u_1^2\rvsim g_1^2/(g_1^2+\|\tilde g\|_2^2)$, where $(g_1,\tilde g)\dsim\RN(0,I_k)$. The moments of $X\dsim\op{Beta}(\alpha,\beta)$ are $\E X^r=\prod_{j=0}^{r-1}\frac{\alpha+j}{\alpha+\beta+j}=\frac{(\alpha)_r}{(\alpha+\beta)_r}$, so from \eqref{eqn:m2-zeta} we obtain \eqref{eqn:M2} for $k\ge2$. For $k=1$, $q^2=1$ almost surely and we also obtain \eqref{eqn:M2}.
\end{proof}

\begin{proof}[Proof of Theorem~\ref{thm:moments}(iii)]
We have $M_1(k,n)=(2n-1)!!2^n(k/2)_n$ \cite{ehrenberg2025transition}, so using part (ii) gives
\begin{align*}
m_2(k,n)&=\frac{(2n-1)!!^2}{(2n)!\sum_{r=0}^n\binom{n}{r}^2\frac{(1/2)_r}{(k/2)_r}}
=\frac{1+o(1)}{\sqrt{\pi n}\sum_{r=0}^n\binom{n}{r}^2\frac{(1/2)_r}{(k/2)_r}}.\numberthis\label{eqn:m2-ratio}
\end{align*}
Let $Q_{n,k}:=\sum_{r=0}^n\binom{n}{r}^2\frac{(1/2)_r}{(k/2)_r}$. We need upper and lower bounds. Since $\binom{n}{r}\le\left(\frac{ne}{r}\right)^r$, we have
\begin{align*}
Q_{n,k}\le 1+\sum_{r=1}^n\left(\frac{ne}{r}\right)^{2r}\prod_{j=0}^{r-1}\frac{1/2+j}{k/2+j}
&\le 1+\sum_{r=1}^n\left(\frac{ne}{r}\right)^{2r}\prod_{j=0}^{r-1}\frac{1+2j}{k}\\
&\le1+\sum_{r=1}^n\left(\frac{ne}{r}\right)^{2r}\left(\frac{2r}{k}\right)^r
=1+\sum_{r=1}^n\left(\frac{2n^2e^2}{kr}\right)^r\\
&\le \exp\left(\frac{2e^2n^2}{k}\right),\numberthis\label{eqn:R-upper}
\end{align*}
where for the last inequality we used that $r^r\ge r!$.

For a lower bound, we will lower bound $Q_{n,k}$ by a single large term in the sum. 
Using $(2r-1)!!=(2r)!/(2^r r!)$ and $\binom{n}{r}\ge\frac{n^r}{r^r}$, then for any $1\le r\le n$,
\begin{align*}
Q_{n,k}\ge \frac{n^{2r}}{r^{2r}}\prod_{j=0}^{r-1}\frac{1+2j}{k+2j}\ge \frac{n^{2r}}{r^{2r}}\frac{(2r-1)!!}{(k+2r)^r}\ge \frac{n^{2r}}{r^{2r}}\frac{r^r}{2^r(k+2r)^r}=\left(\frac{n^2}{2r(k+2r)}\right)^r.\numberthis
\end{align*}
If $k\ge n$, take $r=\lfloor n^2/(4k)\rfloor\le n^2/(4k)\le n$, which gives
\begin{align*}
Q_{n,k}&\ge \left(\frac{2k}{k+k/2}\right)^r=e^{c'r}\ge e^{cn^2/k}.
\end{align*}
(Or use $1+n^2/k$ lower bound if $r=0$.)
If $k\le n$, then we can just use that $Q_{n,k}$ is decreasing in $k$, so $Q_{n,k}\ge Q_{n,n}$ and the above bound with $k=n$ gives $Q_{n,k}\ge e^{cn}$. Thus for any $k$,
\begin{align}\label{eqn:R-lower}
Q_{n,k}&\ge e^{c\min(n,n^2/k)}.
\end{align}

To locate the weak anticoncentration transition, recall from \eqref{eqn:m2-ratio} we have
\begin{align}\label{eqn:m2-simplified}
m_2(k,n)\sim\frac{1}{\sqrt{\pi n}}\frac{1}{Q_{n,k}}.
\end{align}
If $k=\Omega(n^2/\log n)$, then $n^2/k=O(\log n)$, and the upper bound \eqref{eqn:R-upper} on $Q_{n,k}$ gives $Q_{n,k}\le n^{O(1)}$, so \eqref{eqn:m2-simplified} gives
\begin{align}
m_2(k,n)&\ge n^{-O(1)}.
\end{align}
Additionally, if $k/n^2\to\infty$, then \eqref{eqn:R-upper} (along with observing $Q_{n,k}\ge1$ by its $r=0$ term) shows $Q_{n,k}\to1$, which gives \eqref{eqn:m2-limit}.

On the other hand, if $k=o(n^2/\log n)$, then $n^2/k=\omega(\log n)$.
The lower bound \eqref{eqn:R-lower} then shows
\begin{align}
m_2(k,n)\le \frac{1}{\sqrt{\pi n}}e^{-\omega(\log n)}=\frac1{\sqrt{\pi}}n^{-\omega(1)},
\end{align}
which gives lack of weak anticoncentration.
\end{proof}

\section{Page curves closed forms}\label{sec:page}

In this section, we prove closed-form expressions for the R\'enyi-$\alpha$ Page curves for $\alpha\ge2$ an integer.
We work with $m$ modes, with a subsystem $\Gamma$ of size $|\Gamma|=k$ (note, this is a different $k$ than used in Section~\ref{sec:m2}).
We prove the closed-form expression for R\'enyi-2 entropy in Theorem~\ref{thm:s2}, and extend this to R\'enyi-$\alpha$ entropy in Corollary~\ref{cor:alpha}. 
Recall the R\'enyi-$\alpha$ entropy for a reduced state $\rho(U)$ corresponding to a subsystem $\Gamma$ is
\begin{align}
S_\alpha(U):=\frac1{1-\alpha}\log\Tr(\rho(U)^\alpha).
\end{align}
For Gaussian states, the R\'enyi-2 entropy $S_2(U)$ takes a particularly convenient form,
\begin{align}
S_2(U)&=\frac12\log\det\sigma(U)=\frac12\Tr\log\sigma(U),
\end{align}
where $\sigma(U)$ is the $2k\times 2k$ covariance matrix for the reduced evolved state $\rho(U)$ corresponding to subsystem $\Gamma$ of size $|\Gamma|=k$.
Using this expression, \cite{iosue2023page} expressed the average-case R\'enyi-2 entropy $\E S_2(U)$ as an infinite series in terms of the matrix $W=\Pi UU^T\Pi\bar U\bar U^T\Pi$ for $\Pi$ the projection onto the subsystem $\Gamma$. Later, \cite{youm2025average} extended this to derive infinite series expressions for $\E S_\alpha(U)$, $\alpha>2$ an integer.

We give two methods to derive closed-form expressions for the R\'enyi-$\alpha$ Page curves, starting from the infinite series expressions in terms of $W$ from \cite{iosue2023page,youm2025average}. One method uses random matrix resolvent results to study $\E \Tr W^\ell$, and one uses direct summation of hypergeometric terms derived in the infinite series in \cite{iosue2023page}. The random matrix method also gives an alternate derivation of the non-leading constant-order term in the R\'enyi-$\alpha$ entropies as $m\to\infty$.

\begin{thm}[R\'enyi-2 Page curve]\label{thm:s2}
Let $U\dsim\mathcal H_m$ be Haar distributed. From an initial state with all modes equally squeezed by squeezing parameter $s>0$, the R\'enyi-2 Page curve is
\begin{multline}\label{eqn:pagecurve-s2}
\lim_{m\to\infty}\frac1m\E S_2(U;s,r)=\frac12\log\left(\frac{\cosh(2s)+\sqrt{1+|1-2r|^2\sinh^2(2s)}}{2}\right)\\
-\frac{|1-2r|}{2}\log\left(\frac{|1-2r|\cosh(2s)+\sqrt{1+|1-2r|^2\sinh^2(2s)}}{1+|1-2r|}\right).
\end{multline}
\end{thm}
\begin{rmk}\label{rmk:constant}
Along sequences with $k=rm\in\Z$, the next leading order (constant) term $\lambda_2(s,r)$ in the expansion $\E S_2(U)=m\beta_2(s,r)-\lambda_2(s,r)+o(1)$, where $\beta_2(s,r)$ is \eqref{eqn:pagecurve-s2}, is known to be \cite{iosue2023page,breakpoint}
\begin{align}\label{eqn:lambda2}
\lambda_2(s,r)&=-\frac18\log(1-4r(1-r)\tanh^2(2s)).
\end{align}
As referenced in \cite{iosue2023page}, this constant order term \eqref{eqn:lambda2} was derived earlier using breakpoint graphs and combinatorial methods, now written in \cite{breakpoint}.
As we will see, the Jacobi resolvent method we use to prove \eqref{eqn:pagecurve-s2} gives an alternate derivation of the constant order term \eqref{eqn:lambda2}.
\end{rmk}

\begin{proof}
For $r \in [0,1]$ rational, consider sequences of $m$ and subsystem sizes such that $k=rm\in\Z$. 
By Haar invariance, we may take the size $k$ subsystem $\Gamma$ to be the first $k$ modes $\{1,\ldots,k\}$. Additionally, by symmetry $r\leftrightarrow 1-r$, we may assume $k\le m/2$.
The infinite series expansion of $\Tr\log\sigma(U)$ from \cite[(A33)]{iosue2023page} provides the expressions
\begin{align}
S_2(U)&=k\log\cosh(2s)+\frac12\Tr\log(1-\tanh^2(2s)W)\label{eqn:s2-trlog}\\
&=\sum_{\ell=1}^\infty\frac{\tanh^{2\ell}(2s)}{2\ell}(k-\Tr W^\ell),\label{eqn:s2-series}
\end{align}
where $W=\Pi UU^T\Pi\bar U\bar U^T\Pi=:F^\dagger F$, with $\Pi$ the projection onto the subsystem $\Gamma$, and $F=\Pi \bar U\bar U^T\Pi$. 
To evaluate $\E\Tr W^\ell$, we observe that $F$, viewed as a $k\times k$ matrix, is distributed as a $k\times k$ submatrix of an $m\times m$ COE (circular orthogonal ensemble) random matrix (since COE random matrices are distributed as $UU^T$ for $U$ Haar random unitary). The eigenvalues $\lambda_1,\ldots,\lambda_k$ of $W$ are the singular values squared of $F$.
The joint eigenvalue density can be obtained from \cite[\S3.8]{forrester2010log}. To this end, decompose
\begin{align*}
UU^T=\begin{pmatrix}r_{(m-k)\times(m-k)}&t^T_{(m-k)\times k}\\t_{k\times(m-k)}&F_{k\times k}\end{pmatrix},
\end{align*}
where the subscripts indicate the matrix sizes, and $k\le m/2$.
The distribution of the singular values $\tilde s_1,\ldots,\tilde s_k$ of $t$, see e.g. \cite[Proposition 3.8.1]{forrester2010log} with $\beta=1$, is proportional to
\begin{align*}
\prod_{j=1}^k\tilde s_j^{m-2k}\prod_{1\le j<\ell\le k}|\tilde s_j^2-\tilde s_\ell^2|.
\end{align*}
The singular values $s_j$ of $F$ satisfy $s_j^2=1-\tilde s_j^2$.
Performing a change of variables $\lambda_j=1-\tilde s_j^2$ to the eigenvalues of $W$ then gives a Jacobi ensemble with $\beta=1$ and joint density
\begin{align}\label{eqn:jacobi}
p(\lambda_1,\ldots,\lambda_k)&\propto\prod_{j=1}^k(1-\lambda_j)^{\frac{m-2k-1}{2}}\prod_{1\le j<\ell\le k}|\lambda_j-\lambda_\ell|.
\end{align}

Evaluating $\E\Tr W^\ell$ is then asking for $\E\sum_{j=1}^k\lambda_j^\ell$, where the expectation is over the Jacobi ensemble \eqref{eqn:jacobi}. This can be done asymptotically in $m$ using random matrix resolvent methods/Stieltjes transform.
The resolvent (or unnormalized Stieltjes transform of the empirical eigenvalue density) is
\begin{align}\label{eqn:resolvent}
R_k(z)&:=\E\sum_{j=1}^k\frac1{z-\lambda_j}=\sum_{\ell=0}^\infty\frac{\E\Tr W^\ell}{z^{\ell+1}},
\end{align}
where the first expectation value is taken over the Jacobi ensemble, and the second equality is as $z\to\infty$. Therefore, knowing the asymptotic expansion of $R_k(z)$ in $k$ and matching powers of $z$ will give us the asymptotic expansions of $\E\Tr W^\ell$ in $k=rm$.
The linear and constant terms in the $k\to\infty$ expansion of $R_k(z)$ for the Jacobi ensembles are computed in \cite{forrester2017large} and \cite{mezzadri2012moments}. In general, the Jacobi $\beta$-ensemble with parameters $\alpha_1$ and $\alpha_2$ is given by the probability density
\begin{align*}
\frac1{C_k}\prod_{j=1}^k\lambda_j^{\alpha_1}(1-\lambda_j)^{\alpha_2}\oneb_{0<\lambda_j<1}\prod_{1\le j<\ell\le k}|\lambda_j-\lambda_\ell|^\beta.
\end{align*}
The resolvent $R_k^{\alpha_1,\alpha_2}$ can be expanded as a series in $1/k$ as $k\to\infty$ \cite{borot2013asymptotic},
\begin{align}\label{eqn:rk-series}
R^{\alpha_1,\alpha_2}_k(z)&=\sum_{\ell=0}^\infty\frac{2^\ell a_\ell(z)}{k^{\ell-1}}.
\end{align}
\begin{prop}[special case of {\cite[Proposition 4.7]{forrester2017large}}]\label{prop:wterms}
Consider the Jacobi $\beta$-ensemble with $\beta=1$ and parameters $\alpha_1=\Theta(k)$ and $\alpha_2=0$.
Write $\alpha_1=\tilde\alpha_1k/2+\delta_1$ with $\tilde\alpha_1,\delta_1=O(1)$, and let $c_-=\tilde\alpha_1^2/(\tilde\alpha_1+2)^2$. Then the first two coefficients $a_\ell(z)$ in \eqref{eqn:rk-series} are
\begin{align*}
a_0(z)&=-\frac{\tilde\alpha_1}{2z}+\frac{\tilde\alpha_1+2}{2z(z-1)}\sqrt{(z-1)(z-c_-)},\\
a_1(z)&=-\frac{1-c_-}{8(z-1)(z-c_-)}-(1/2+\delta_1)\left(\frac{1}{2z}+\frac{\tilde\alpha_1-(\tilde\alpha_1+2)z}{2(\tilde\alpha_1+2)z\sqrt{(z-1)(z-c_-)}}\right).
\end{align*}
\end{prop}

For the Jacobi ensemble \eqref{eqn:jacobi}, we have $\alpha_1=0$ and $\alpha_2=(\frac{1}{2r}-1)k-\frac12$; this is dual to $\alpha_1=(\frac{1}{2r}-1)k-\frac12$ and $\alpha_2=0$ under the change of variables $\lambda_j\mapsto 1-\lambda_j$. The map $\lambda_j\mapsto1-\lambda_j$ in $R_k(z)$ is equivalent to sending $R_k(z)\mapsto -R_k(1-z)$. 
Applying this map with $\tilde\alpha_1=\frac1r-2$, $\delta_1=-1/2$, and $c_-=(1-2r)^2$ gives
\begin{align}
a_0(z)&=\frac{\frac1r-2}{2(1-z)}-\frac{\sqrt{-z(1-z-(1-2r)^2)}}{2r(1-z)z}=\frac{-(1-2r)+\sqrt{1-4r(1-r)/z}}{2r(z-1)},\\
a_1(z)&=-\frac{1-(1-2r)^2}{8z(1-z-(1-2r)^2)}=\frac{R}{2z(z-4R)},
\end{align}
where $R:=r(1-r)$. (When $r=1/2$ we may also use \cite[Proposition 4.6]{forrester2017large}.)
This gives the resolvent expansion
\begin{align}
R_k(z)&=km_0(z)+m_1(z)+O(1/k),
\end{align}
with
\begin{align}\label{eqn:m}
m_0(z)&=\frac{-(1-2r)+\sqrt{1-4r(1-r)/z}}{2r(z-1)},\quad m_1(z)=\frac14\left(\frac1{z-4R}-\frac1z\right).
\end{align}

For the Page curve formula \eqref{eqn:pagecurve-s2}, from \eqref{eqn:s2-trlog}, we want to compute, for $x:=\tanh^2(2s)$,
\begin{align*}
\beta_2(s,r)&=r\log\cosh(2s)+\lim_{k\to\infty}\frac rk\E\left[\frac1{2}\Tr\log(1-xW)\right]\\
&=r\log\cosh(2s)+\frac r2\int_0^1 \log(1-x\lambda)\,d\mu_r(\lambda),\numberthis
\end{align*}
where $\mu_r$ is the limiting empirical eigenvalue distribution $\lim_{k\to\infty}\E\frac1k\sum_{j=1}^k\delta_{\lambda_j}$, whose Stieltjes transform is $\int\frac{1}{z-\lambda}\,d\mu_r(\lambda)=m_0(z)$.
We could use Stieltjes inversion to determine $\mu_r$ and integrate, but it will be a bit faster to differentiate instead. Since $\log\cosh(2s)=-\frac12\log(1-x)$, then
\begin{align*}
\partial_x\beta_2&=\frac{r}{2(1-x)}-\frac r2\int_0^1\frac{\lambda}{1-x\lambda}\,d\mu_r(\lambda)\\
&=\frac{r}{2(1-x)}-\frac{r}{2x}\int_0^1\frac{1}{1-x\lambda}\,d\mu_r(\lambda)+\frac{r}{2x}\int_0^11\,d\mu_r(\lambda).\numberthis
\end{align*}
Since $\int_0^1\frac{1}{1-x\lambda}\,d\mu_r(\lambda)=\frac1xm_0(1/x)$, we have
\begin{align}\label{eqn:deriv-beta}
\partial_x\beta_2&=\frac{r}{2(1-x)}-\frac{r}{2x}\left[\frac 1xm_0(1/x)-1\right]
=\frac{-1+\sqrt{1-4r(1-r)x}}{4x(x-1)}.
\end{align}
This has a closed-form antiderivative (set for example $u=\sqrt{1-4r(1-r)x}$ or use Mathematica)
\begin{multline}
\frac14\bigg[2(r-1)\log\Big(1-2r+\sqrt{1-4r(1-r)x}\Big)+2\log\Big(1+\sqrt{1-4r(1-r)x}\Big)\\
-2r\log\left(-1+2r+\sqrt{1-4r(1-r)x}\right)\bigg],
\end{multline}
so integrating \eqref{eqn:deriv-beta} from $0$ to $x$ gives
\begin{align}\label{eqn:limit3}
\beta_2(s,r)&=\frac12(2r-1)\log\left(\frac{1-2r+\sqrt{1-4r(1-r)x}}{2-2r}\right)+\frac12\log\left(\frac{1+\sqrt{1-4r(1-r)x}}{2}\right)-\frac{r}{2}\log(1-x).
\end{align}
To obtain the form \eqref{eqn:pagecurve-s2}, recall $r\le1/2$ and $x=\tanh^2(2s)$, and use that $1=\sech^2(2s)+\tanh^2(2s)$ to note that 
$\sqrt{1-4r(1-r)x}=\sqrt{\sech^2(2s)+|1-2r|^2\tanh^2(2s)}=\sech(2s)\sqrt{1+|1-2r|^2\sinh^2(2s)}$.

Next, we derive the constant order term \eqref{eqn:lambda2}. Expanding $m_1$ in \eqref{eqn:m} as a power series in $1/z$, we have
\begin{align}
m_1(z)=\frac14\left(\frac1{z-4R}-\frac1z\right)&=\sum_{\ell=1}^\infty\frac{4^{\ell-1}R^\ell}{z^{\ell+1}}.
\end{align}
Matching powers of $z$ with \eqref{eqn:resolvent} and \eqref{eqn:rk-series} gives
\begin{align}
\E\Tr W^\ell&=mf_\ell(r)+4^{\ell-1}[r(1-r)]^\ell+o(1),
\end{align}
and so
\begin{align*}
\lambda_2(s,r)&=\sum_{\ell=1}^\infty\frac{\tanh^{2\ell}(2s)}{2\ell}4^{\ell-1}[r(1-r)]^\ell\\
&=-\frac18\log(1-4r(1-r)\tanh^2(2s)),\numberthis
\end{align*}
as desired.
\end{proof}

\begin{proof}[Proof (alternate proof of \eqref{eqn:pagecurve-s2} in Theorem~\ref{thm:s2})]
\cite[Theorem 2]{iosue2023page} gives the infinite series expansion
\begin{align*}
\lim_{m\to\infty}\frac1m\E S_2(U)&=r\log\cosh(2s)-\sum_{\ell=1}^\infty r^{\ell+1}\tanh^{2\ell}(2s)\frac{C_\ell}{2\ell}{}_2F_1(1-\ell,\ell;\ell+2;r)\\
&=\sum_{\ell=1}^\infty \frac{\tanh^{2\ell}(2s)}{2\ell}[r-r^{\ell+1}C_\ell\,{}_2F_1(1-\ell,\ell;\ell+2;r)],\numberthis
\end{align*}
where $C_\ell:=\frac1{\ell+1}\binom{2\ell}{\ell}$ is the $\ell$th Catalan number, and ${}_2F_1$ is the hypergeometric function. We will evaluate this series.

Let $G_\ell(r):=r-f_\ell(r)$ for $f_\ell(r)=r^{\ell+1}C_\ell\,{}_2F_1(1-\ell,\ell;\ell+2;r)$. Due to \cite[Lemmas C.1,C.2]{iosue2023page}, we know that each $f_\ell$ is a polynomial in $r$ with only terms of degree in $\{\ell+1,\cdots,2\ell\}$. (An alternate way to see this now that we start with the explicit formula for $f_\ell$ is to observe that since the first entry is an integer $1-\ell\le0$, the hypergeometric function is a polynomial of degree at most $\ell-1$.)
This implies that $G_\ell(r)=r+O(r^{\ell+1})$ as $r\to0$.
Also, as shown in the proof of \cite[Lemma C.2]{iosue2023page}, each individual $G_\ell(r)$ must be symmetric under $r\mapsto 1-r$ (this follows physically from the symmetry of the Page curve for pure states). 

Using these two properties, we can determine the polynomial $G_\ell$. By the symmetry property, $G_\ell$ must be a polynomial in $R:=r(1-r)$ of degree at most $\ell$; say $G_\ell(r)=P_\ell(R)$. Since $P_\ell(R)=G_\ell(r)=r+O(r^{\ell+1})$, we invert $R=r-r^2$  near $R=0$ to see
\begin{align}\label{eqn:r-catalan}
r=\frac{1-\sqrt{1-4R}}{2}&=\sum_{j=1}^\infty C_{j-1}R^j,
\end{align}
which is essentially the Catalan generating function. Then the only degree $\ell$ polynomial in $R$ that agrees with $r$ to order $O(r^\ell)=O(R^\ell)$ is $\sum_{j=1}^\ell C_{j-1}R^j$, so this must be $P_\ell(R)$. So we get
\begin{align}\label{eqn:Gl}
G_\ell(r)=\sum_{j=1}^\ell C_{j-1}[r(1-r)]^j.
\end{align}

From \eqref{eqn:Gl}, we see
\begin{align}\label{eqn:Gsum}
\beta_2(s,r):=\lim_{m\to\infty}\frac1m\E S_2(U)&=\sum_{\ell=1}^\infty\frac{\tanh^{2\ell}(2s)}{2\ell}\sum_{j=1}^\ell C_{j-1}[r(1-r)]^j.
\end{align}
Note that $0\le G_\ell(r)\le 1/2$ using \eqref{eqn:r-catalan}, so \eqref{eqn:Gsum} is uniformly convergent for $x:=\tanh^{2}(2s)<1$, and we can differentiate term-by-term to obtain
\begin{align*}
\partial_x\beta_2=\frac12\sum_{\ell=1}^\infty x^{\ell-1}\sum_{j=1}^\ell C_{j-1}R^j
&=\frac12\sum_{j=1}^\infty C_{j-1}R^j\sum_{\ell=j}^\infty x^{\ell-1}\\
&=\frac12\sum_{j=1}^\infty C_{j-1}\frac{(Rx)^{j}x^{-1}}{1-x}
=\frac{1-\sqrt{1-4xR}}{4x(1-x)},\numberthis
\end{align*}
using the Catalan generating function \eqref{eqn:r-catalan} for the last equality.
This is \eqref{eqn:deriv-beta}, which gives the closed form \eqref{eqn:pagecurve-s2}.
\end{proof}

The following analogous results for the R\'enyi-$\alpha$ entropy for $\alpha\ge2$ an integer will follow from the $\alpha=2$ case and \cite[Theorem 4]{youm2025average}. The von Neumann entropy ($\alpha=1$) infinite series expression in \cite{youm2025average} could also likely be simplified, though perhaps not to the same extent as the R\'enyi-$\alpha$ Page curve below.
\begin{cor}[R\'enyi-$\alpha$ entropy Page curves]\label{cor:alpha}
Let
\begin{multline}\label{eqn:fr}
f_r(z):=-\frac{1-|1-2r|}{4}\log(1-z)+\frac12\log\left(\frac{1+\sqrt{1-4r(1-r)z}}{2}\right)\\ 
-\frac{|1-2r|}{2}\log\left(\frac{|1-2r|+\sqrt{1-4r(1-r)z}}{1+|1-2r|}\right).
\end{multline}
For any integer $\alpha\ge2$, the R\'enyi-$\alpha$ Page curve is
\begin{align}\label{eqn:pagecurve-alpha}
\lim_{m\to\infty}\frac1m\E S_\alpha(U)&=\frac{1}{\alpha-1}\left[\sum_{j=1}^{\lfloor\frac{\alpha-1}{2}\rfloor}2f_r(z_j)+\oneb_{\alpha\in2\Z}f_r(\tanh^2(2s))\right],
\end{align}
where
\begin{align}
z_j:=\frac{\sinh^2(2s)}{\cosh^2(2s)+\cot^2(\frac{\pi j}{\alpha})}.
\end{align}
\end{cor}

\begin{rmk}
Similar to Remark~\ref{rmk:constant}, we can derive the next subleading (constant) order term $\lambda_\alpha(s,r)$ in the expansion $\E S_\alpha(U)=m\beta_\alpha(s,r)-\lambda_\alpha(s,r)+o(1)$, where $\beta_\alpha(s,r)$ is \eqref{eqn:pagecurve-alpha}.
Letting $h_\alpha(\nu):=\frac{1}{\alpha-1}\log\left((\nu+1)^\alpha-(\nu-1)^\alpha\right)$, then
\begin{align}\label{eqn:lambda-alpha}
\lambda_\alpha(s,r)&=\frac1{4}\left[h_\alpha(\cosh(2s))-h_\alpha\left(\sqrt{1+|1-2r|^2\sinh^2(2s)}\right)\right].
\end{align}
\end{rmk}

\begin{proof}
Note that \eqref{eqn:fr} is just \eqref{eqn:limit3} written for any $0\le r\le1$. Theorem 4 of \cite{youm2025average} gives the formula
\begin{align}
\E S_\alpha(U)&=\frac{\oneb_{\alpha\in2\Z}}{\alpha-1}\E S_2(U)+\frac1{\alpha-1}\sum_{j=1}^{\lfloor\frac{\alpha-1}{2}\rfloor}\sum_{\ell=1}^\infty \frac{\sinh^{2\ell}(2s)}{\ell(\cosh^2(2s)+\cot^2\frac{\pi j}{\alpha})^\ell}(mG_\ell(r)-H_\ell(r)+o(1)),
\end{align}
where $G_\ell(r)=r-r^{\ell+1}C_\ell{}_2F_1(1-\ell,\ell,2+\ell,r)$ and $H_\ell(r)=4^{\ell-1}(r(1-r))^\ell$ are same as in the R\'enyi-2 definitions. Equation~\eqref{eqn:pagecurve-alpha} for $\beta_\alpha(s,r)$ then immediately follows from the R\'enyi-2 calculations for Theorem~\ref{thm:s2}.

For the constant term \eqref{eqn:lambda-alpha}, let $L_r(z):=\sum_{\ell=1}^\infty\frac{z^\ell}{2\ell}H_\ell(r)=-\frac18\log(1-4zR)$, where $R=r(1-r)$. Then
\begin{align}\label{eqn:lambda-L}
\lambda_\alpha(s,r)&=\frac{1}{\alpha-1}\left[\sum_{j=1}^{\lfloor\frac{\alpha-1}{2}\rfloor}2L_r(z_j)+\oneb_{\alpha\in2\Z}L_r(\tanh^2(2s))\right].
\end{align}
To simplify this, let $\Delta_{r,s}:=\sqrt{1+|1-2r|^2\sinh^2(2s)}=\sqrt{\cosh^2(2s)-4R\sinh^2(2s)}$, and note that
\begin{align}
1-4z_jR&=1-\frac{4R\sinh^2(2s)}{\cosh^2(2s)+\cot^2(\frac{\pi j}{\alpha})}
=\frac{\Delta_{r,s}^2+\cot^2(\frac{\pi j}{\alpha})}{\cosh^2(2s)+\cot^2(\frac{\pi j}{\alpha})}.
\end{align}
Also, $1-4R\tanh^2(2s)=1-\frac{4R\sinh^2(2s)}{\cosh^2(2s)}=\frac{\Delta_{r,s}^2}{\cosh^2(2s)}$.
From \eqref{eqn:lambda-L}, we then get
\begin{align}\label{eqn:cotangent}
\lambda_\alpha(s,r)&=\frac1{4(\alpha-1)}\left[\log\frac{\prod_{j=1}^{\lfloor\frac{\alpha-1}{2}\rfloor}(\cosh^2(2s)+\cot^2(\frac{\pi j}{\alpha}))}{\prod_{j=1}^{\lfloor\frac{\alpha-1}{2}\rfloor}(\Delta_{r,s}^2+\cot^2(\frac{\pi j}{\alpha}))}+\oneb_{\alpha\in2\Z}\log\frac{\cosh(2s)}{\Delta_{r,s}}\right].
\end{align}

The polynomial $p_\alpha(z):=(z+1)^\alpha-(z-1)^\alpha$ is degree $\alpha-1$ and has roots where $\left(\frac{z+1}{z-1}\right)^\alpha=1$, which gives roots for $j=1,\ldots,\alpha-1$
\begin{align}
\frac{z+1}{z-1}=e^{2\pi ij/\alpha},\quad\Rightarrow\quad z=\frac{e^{2\pi ij/\alpha}+1}{e^{2\pi ij/\alpha}-1}=-i\cot\frac{\pi j}{\alpha}.\label{eqn:cotangent-alpha}
\end{align}
If $\alpha\in2\Z$, this includes the root $z=0$ when $j=\alpha/2$. For $j\ne\alpha/2$, the roots in \eqref{eqn:cotangent-alpha} come in pairs $\pm i\cot(\pi j/\alpha)$.
Thus for some leading coefficient $C_\alpha$ which will cancel in the ratios in \eqref{eqn:cotangent},
\begin{align}
z^{\oneb_{\alpha\in2\Z}}\prod_{j=1}^{\lfloor\frac{\alpha-1}{2}\rfloor}\left(z^2+\cot^2\frac{\pi j}{\alpha}\right)&=C_\alpha\left((z+1)^\alpha-(z-1)^\alpha\right).
\end{align}
Using this in \eqref{eqn:cotangent} gives \eqref{eqn:lambda-alpha}.
\end{proof}

\section{Entanglement with unequal squeezing}\label{sec:unequal}

In Section~\ref{sec:page}, as well as in \cite{iosue2023page,youm2025average,shou2026entanglement}, the initial state has all modes equally squeezed with some squeezing parameter $s>0$. 
In this section, we consider unequal squeezing parameters $\mathbf s=(s_1,\ldots,s_m)$, and show that one can often bound the expected R\'enyi-2 entropy in terms of the R\'enyi-2 entropy for an initial state with all input modes squeezed. Thus the previously mentioned results for equally squeezed input states also extend to give bounds in the unequal squeezing case.
Throughout, by increasing we will mean $x\le y$ implies $f(x)\le f(y)$; in other words, we mean nondecreasing, and do not mean strictly increasing. 
In this section, we work with $m$ modes, with a subsystem $\Gamma$ of size $|\Gamma|=k$ (note, this is the same $k$ as used in Section~\ref{sec:page}, but a different $k$ than used in Section~\ref{sec:m2}).

\begin{thm}[monotonicity]\label{thm:monotonicity}
Let $\sq=(s_1,\ldots,s_m)\in\R^m$, and let $\Gamma\subset\{1,\ldots,m\}$. Denote by $S_2(U,\sq)$ the R\'enyi-$2$ entropy for subsystem $\Gamma$, started from an initial Gaussian state with squeezing parameters $\sq$ and evolved by the linear optical unitary $U$. 
For any random ensemble of $U$ whose distribution is invariant under right multiplication by phase shifts by $\pi/2$, i.e.~multiplication by diagonal matrices with $\{\pm i,\pm1\}$ entries on the diagonal, the expected values
$\E S_2(U,\sq)$ and $\E[S_2(U,\sq)^2]$ are increasing in each $|s_i|$, $i=1,\ldots,m$. More generally, for any convex increasing $\Phi:\R_{\ge0}\to\R$, $\E\Phi(S_2(U,\sq))$ is increasing in each $|s_i|$.
\end{thm}
This proves a generalization of the conjecture \cite[Conjecture 10]{iosue2023page}. Note it also applies to non-Haar random ensembles such as the depth-$d$ random linear optical networks studied in \cite{shou2026entanglement}.
As noted in \cite{iosue2023page}, the non-averaged quantity $S_2(U,\sq)$ need not satisfy a monotonicity property in the squeezing. A specific example is when $U$ is a real orthogonal matrix. If all modes are equally squeezed with squeezing parameter $s>0$, then $S_2(U,s)=0$ since $UU^T=I$, and the covariance matrix is unchanged under evolution by $U$. However, one can obtain nonzero entanglement for example by setting some of the squeezing parameters $s$ to $0$. Thus the averaging over the random unitary $U$ is crucial for the theorem.

As an immediate application of the theorem, let $s_\mathrm{min}:=\min_i|s_i|$ and $s_\mathrm{max}:=\max_i|s_i|$. For scalar $s$, let $S_2(U,s):=S_2(U,(s,\ldots,s))$. Then for any ensemble as in the theorem,
\begin{align}\label{eqn:s-compare}
\E S_2(U,s_\mathrm{min})\le \E S_2(U,\sq)&\le \E S_2(U,s_\mathrm{max}).
\end{align}

\begin{proof}[Proof of Theorem~\ref{thm:monotonicity}]
We first show that for any fixed unitary $U$, the map $\sq\mapsto S_2(U,\sq)$ is convex. This will follow from the formula (see e.g. \cite{Fukuda2019Typical-entangl,Serafini2017Quantum-Continu,iosue2023page}) for the R\'enyi-2 entropy $S_2(U)=\frac12\log\det\sigma(U)$ in terms of the reduced evolved covariance matrix
\begin{align}
\sigma(U)=\hat P \eta(U)\sigma_0(\sq)\eta(U)^T\hat P^T,
\end{align}
where $\sigma_0(\sq)=\operatorname{diag}(e^{2s_1},\ldots,e^{2s_m},e^{-2s_1},\ldots,e^{-2s_m})$ is the initial covariance matrix, $\eta(U)=\begin{pmatrix}\Re U&\Im U\\-\Im U&\Re U\end{pmatrix}$ is the usual embedding of $\U(m)$ into $\Sp(2m)\cap O(2m)$, and $\hat P=P\oplus P$ for $P:\C^m\to\C^\Gamma$ the projection onto $\Gamma$.
Letting $A(U):=\hat P\eta(U)$ which is $2k\times 2m$, we then consider $\det(A(U)\sigma_0(\sq)A(U)^T)$. This can be evaluated using the Cauchy--Binet formula (which generalizes the relation $\det(AB)=\det(BA)$ for square matrices to a formula for rectangular matrices). Letting $A_I$ denote the $2k\times|I|$ matrix consisting of the columns of $A$ indexed by $I$, and recalling that $\sigma_0(\sq)$ is diagonal, applying the Cauchy--Binet formula gives
\begin{align*}
S_2(U,\sq)&=\frac12\log\det(A(U)\sigma_0(\sq)A(U)^T)
=\frac12\log\sum_{\substack{I\subset[2m]\\|I|=2k}} \det(A_{I}(U))^2 e^{2\langle \xi_I,\sq\rangle},\numberthis\label{eqn:logdetsum}
\end{align*}
where $\xi_I\in\R^m$ is defined via $(\xi_I)_j=\oneb_I[j]-\oneb_I[j+m]$.
Each term in the sum in \eqref{eqn:logdetsum} is a nonnegative exponential in $\sq$ so is log-convex in $\sq$, and the sum of log-convex functions is log-convex, so $\sq\mapsto S_2(U,\sq)$ is convex.

Additionally, $\E S_2(U,\sq)$ is invariant under $s_i\mapsto -s_i$ by the random ensemble property. So each map $s_i\mapsto \E S_2(U,\sq)$ is convex and even, and so must be increasing in $|s_i|$.
For any function $\Phi:\R_{\ge0}\to\R$ which is convex and increasing, $\Phi\circ S_2(U,\cdot)$ is also convex, and so by the same argument as above, the map $s_i\mapsto \E\Phi(S_2(U,\sq))$ is then also increasing in $|s_i|$.
\end{proof}

The lower bound in \eqref{eqn:s-compare} is not useful when some of the modes are vacuum ($s_i=0$).
We consider the setup where $L$ input modes are squeezed with squeezing parameter $s>0$, and the rest are vacuum, and prove the following bounds for Haar random $U$.
Note that the small squeezing case for arbitrary $\sq$ where $s_\mathrm{max}=o(1)$ is already covered in \cite[Corollary 9]{iosue2023page}, so we focus on bounds which are useful for fixed $s$ (though they apply to arbitrary $s$). 
When clear from context, we let $s\oneb_{L}$ denote both the length $m$ vector consisting of $s$ repeated $L$ times and $0$ repeated $m-L$ times, and the length $L$ vector consisting of $s$ repeated $L$ times.

\begin{thm}[$L$ modes equally squeezed]\label{thm:unequal}
Let $U\dsim\mathcal H_m$ be Haar random, and consider an initial Gaussian state with $L$ modes squeezed with squeezing parameter $s>0$, and the remaining $m-L$ modes vacuum. Denote by $S_2(U,s\oneb_L)$ the R\'enyi-2 entropy for a size $k$ subsystem starting from such an initial state. Then
\begin{align}\label{eqn:lmodes}
\E S_2\bigg(U,\frac{L}{m}s\bigg)\le \E S_2(U,s\oneb_{L})\le \E S_2(U,s).
\end{align}
Also,
\begin{align}\label{eqn:lmodes2}
\frac12e^{-4s}\sinh^2(2s)\frac{Lk(m-k)}{m(m+1)}\le \E S_2(U,s\oneb_{L}) \le \min(k,m-k)\log\bigg(1+\frac{L}{m}(\cosh(2s)-1)\bigg).
\end{align}
In particular, letting $r:=k/m$, then \eqref{eqn:lmodes2} implies $\E S_2(U,s\oneb_L)=\Theta_s(\min(r,1-r))L$. Note that for large $s$ and $L\propto m$, the lower bound in \eqref{eqn:lmodes} can give a better (linear) $s$ dependence.
\end{thm}
\begin{proof}
The upper bound in \eqref{eqn:lmodes} follows from the previous monotonicity theorem. The lower bound also follows from the result that $\sq\mapsto S_2(U,\sq)$ is convex. By Haar-invariance of $U$ over permutations of the locations of the $L$ initial squeezed modes, followed by Jensen's inequality, we have
\begin{align*}
\E_U S_2(U,s\oneb_L)&=\E_U\E_{\sigma\in S_m}S_2(U,\sigma(s\oneb_L,\mathbf 0_{m-L}))\\
&\ge \E_US_2(U,\E_{\sigma\in S_m}\sigma(s\oneb_L,\mathbf 0_{m-L}))
=\E_US_2\bigg(U,\frac{L}{m}s\bigg).\numberthis
\end{align*}

The upper bound in \eqref{eqn:lmodes2} follows from the application of Jensen's inequality in the proof of \cite[Corollary 9]{iosue2023page}. Let $\nu=\frac1m\sum_i \cosh(2s_i)=\frac{L}{m}\cosh(2s)+1-\frac{L}{m}$; then \cite[Eq.~(A5)]{iosue2023page} shows $\E\sigma(U)=\nu I_{2k}$. Since $X\mapsto\log\det X$ is concave on positive-definite matrices \cite[\S3.1.5]{boyd2004convex}, then for $k\le m/2$, Jensen's inequality gives
\begin{align}
\E S_2(U,s\oneb_L)&\le \frac12\Tr\log\E\sigma(U)
=k\log\bigg(1+\frac{L}{m}(\cosh(2s)-1)\bigg).
\end{align}

For the lower bound in \eqref{eqn:lmodes2}, by symmetry $r\leftrightarrow 1-r$ we may consider a subsystem $\Gamma$ with $|\Gamma|=k\le m/2$.
Let $\tilde \sigma(U)=\eta(U)\sigma_0\eta(U)^T$ be the (unreduced) covariance matrix for the evolved initial state under $U$, and write it in block form across $\Gamma\oplus\Gamma^c$,
\begin{align*}
\tilde\sigma(U)=\begin{pmatrix}A_\Gamma&C\\ C^T&A_{\Gamma^c}\end{pmatrix},
\end{align*}
for $A_\Gamma:=\hat P_{\Gamma}\tilde\sigma(U)\hat P_{\Gamma}^T$, $A_{\Gamma^c}:=\hat P_{\Gamma^c}\tilde\sigma(U)\hat P_{\Gamma^c}^T$, and $C:=\hat P_\Gamma\tilde\sigma(U)\hat P_{\Gamma^c}^T$, where $\hat P_\Gamma=P_\Gamma\oplus P_\Gamma$ for $P_\Gamma:\C^m\to\C^\Gamma$ the projection onto $\Gamma$, and similarly for $\hat P_{\Gamma^c}$.
Letting $Q_L$ be the $m\times m$ projection matrix onto the $L$ squeezed input modes and using the covariance matrix expressions in \cite[\S A.1]{iosue2023page}, we see for example that
\begin{equation}\label{eqn:C}
    C= (\cosh(2s)-1)M_{1} + \sinh(2s)M_{2},
\end{equation}
where 
\begin{equation}
    M_{1} = \begin{pmatrix}
            P_\Gamma\Re[UQ_LU^{\dag}]P_{\Gamma^c}^T & P_\Gamma\Im[UQ_LU^{\dag}]P_{\Gamma^c}^T \\
             -P_\Gamma\Im[UQ_LU^{\dag}]P_{\Gamma^c}^T & P_\Gamma\Re[UQ_LU^{\dag}]P_{\Gamma^c}^T 
        \end{pmatrix},
\end{equation}
and
\begin{equation}
    M_{2} = \begin{pmatrix}
            P_\Gamma\Re[\bar{U}Q_LU^{\dag}]P_{\Gamma^c}^T & P_\Gamma\Im[\bar{U}Q_LU^{\dag}]P_{\Gamma^c}^T \\
            P_\Gamma\Im[\bar{U}Q_LU^{\dag}]P_{\Gamma^c}^T & -P_\Gamma\Re[\bar{U}Q_LU^{\dag}]P_{\Gamma^c}^T 
        \end{pmatrix}.
\end{equation}

By Schur complement, $1=\det\tilde\sigma(U)=\det(A_\Gamma)\det(A_{\Gamma^c})\det(I-A_{\Gamma^c}^{-1/2}C^TA_\Gamma^{-1}CA_{\Gamma^c}^{-1/2})$.
Using that $\det(A_\Gamma)=\det(A_{\Gamma^c})$ since the state is pure, we then have
\begin{align}
S_2(U,s\oneb_L)=\frac12\log\det A_\Gamma
&=-\frac14\log\det(I-A_{\Gamma^c}^{-1/2}C^TA_\Gamma^{-1}CA_{\Gamma^c}^{-1/2}).
\end{align}
$A_\Gamma$ and $A_{\Gamma^c}$ are symmetric, so $Y:=A_{\Gamma^c}^{-1/2}C^TA_\Gamma^{-1}CA_{\Gamma^c}^{-1/2}=(A_\Gamma^{-1/2}CA_{\Gamma^c}^{-1/2})^T(A_\Gamma^{-1/2}CA_{\Gamma^c}^{-1/2})\ge0$.
Also, $A_\Gamma^{-1},A_{\Gamma^c}^{-1}\ge e^{-2s}$, since $A_\Gamma,A_{\Gamma^c}\le e^{2s}$ as $\tilde\sigma(U)$ has the same spectrum as $\sigma_0$.
Since $\tilde\sigma(U)>0$ and $A_\Gamma>0$, the Schur complement and $1-Y$ are also $>0$. 
Then $0\le Y<1$, and using that $-\log(1-y)\ge y$ for $0\le y\le1$, we obtain
\begin{align*}
S_2(U,s\oneb_L)&\ge\frac14\Tr(A_{\Gamma^c}^{-1/2}C^TA_\Gamma^{-1}CA_{\Gamma^c}^{-1/2})
\ge\frac14e^{-4s}\|C\|_\hs^2.\numberthis\label{eqn:s2c}
\end{align*}

Next, using \eqref{eqn:C}, we compute
\begin{align}\label{eqn:cm12}
\|C\|_\hs^2&=(\cosh(2s)-1)^2\|M_1\|_\hs^2+\sinh^2(2s)\|M_2\|_\hs^2,
\end{align}
since one can check that $\langle M_1,M_2\rangle_\hs=\Tr(M_1^TM_2)=0$.
We have $\|M_2\|_\hs^2=2\|P_\Gamma UQ_LU^T P_{\Gamma^c}^T\|_\hs^2$, so
\begin{align*}
\E\|M_2\|_\hs^2&=2\sum_{x\in\Gamma}\sum_{y\in\Gamma^c}\E|\langle x|UQ_LU^T|y\rangle|^2\\
&=2\sum_{x\in\Gamma}\sum_{y\in\Gamma^c}\sum_{\ell,n=1}^L\E U_{x\ell}U_{y\ell}\bar U_{xn}\bar U_{yn}
=\frac{2Lk(m-k)}{m(m+1)},\numberthis
\end{align*}
using Weingarten calculus \cite{CollinsSniady2006} for the fourth moment in the last line.
Thus with \eqref{eqn:s2c} and \eqref{eqn:cm12}, we get the lower bound in \eqref{eqn:lmodes2}. Note while we could also include the contribution from $\|M_1\|_\hs^2=2\|P_\Gamma UQ_LU^\dagger P_{\Gamma^c}^T\|_\hs^2$, it would give at most a small constant factor improvement.

The asymptotic $\E S_2(U,s\oneb_L)=\Theta_s(\min(r,1-r)L)$ follows from \eqref{eqn:lmodes2}. For large $s$, we can use the lower bound in \eqref{eqn:lmodes} as follows. Note that $k-\E\Tr W=k-\E\|F\|_\hs^2=\Theta(k(m-k)/m)=\Theta(\min(k,m-k))$. Then since $\Tr W^\ell\le \Tr W$ as $W$ has eigenvalues in $[0,1]$, the series expansion \eqref{eqn:s2-series} gives
\begin{align*}
\E S_2\bigg(U,\frac{L}{m}s\bigg)&\ge (k-\E\Tr W)\log\cosh\left(\frac{2Ls}{m}\right)\\
&\ge\Theta(\min(k,m-k))\log\cosh\left(\frac{2Ls}{m}\right)\\
&\ge\Theta(\min(r,1-r))Ls,\numberthis
\end{align*}
for example if $\frac{Ls}{m}-\log 2>0$ so that $\log\cosh(2Ls/m)\ge \log(\frac12e^{2Ls/m})=\Theta(Ls/m)$.
\end{proof}

\vspace{2mm}
\noindent
\textbf{Acknowledgments.}
This project used GPT-5.5 Thinking and Pro for coming up with proof ideas and methods, as well as for general checking and proofreading. The paper was written by the authors and all results and proofs were checked and validated by the authors, who are fully responsible for the final content. L.S., J.T.I, A.E., and A.V.G.\ acknowledge support from the U.S.~Department of Energy, Office of Science, Accelerated Research in Quantum Computing, Fundamental Algorithmic Research toward Quantum Utility (FAR-Qu). L.S., J.T.I, A.E., and A.V.G.\ were also supported in part by ARL (W911NF-24-2-0107), ONR MURI, NSF QLCI (award No.~OMA-2120757), NQVL:QSTD:Design:FTL, DoE ASCR Quantum Testbed Pathfinder program (award No.~DE-SC0024220), NSF STAQ program, and AFOSR MURI. L.S., J.T.I, A.E., and A.V.G.\ also acknowledge support from the U.S.~Department of Energy, Office of Science, National Quantum Information Science Research Centers, Quantum Systems Accelerator (award No.~DE-SCL0000121).
Y.-X.W.~acknowledges support from a QuICS Hartree Postdoctoral Fellowship. 
J.T.I's contributions were made during his time at UMD.

\vspace{2mm}
\noindent
\textit{Note added.} During the final stages of completing this manuscript, we became aware of independent work by Zhao \cite{zhao2026exact}, which also derives a closed form for the second moments $M_2(k,n)$ and weak anticoncentration transition at $k\propto n^2/\log n$.

\bibliographystyle{amsnoalpha_edit}
\bibliography{gbs.bib}

\end{document}